\documentclass[journal,twoside,web]{ieeecolor}

\usepackage{lcsys}
\usepackage{cite}
\usepackage{amsmath,amssymb,amsfonts}
\usepackage{algorithmic}
\usepackage{graphicx}
\usepackage{textcomp}

\usepackage{centernot}
\usepackage{longtable}
\usepackage{amsmath}
\usepackage{graphicx, amssymb}
\usepackage[dvips]{epsfig}
\usepackage{color}
\usepackage[dvipsnames]{xcolor}
\usepackage{comment}
\usepackage{balance}
\usepackage{tikz}

\usepackage{color}
\usepackage[dvipsnames]{xcolor}

\usepackage{multirow}
\usepackage{lcsys}
\usepackage{cite}
 \usepackage{booktabs}
\usepackage{comment}
\usepackage{amsmath,amssymb,amsfonts}
\usepackage{graphicx}
\usepackage{cite}
\usepackage{mathtools}
\usepackage{bm}
\usepackage{color}
\usepackage{hyperref}
\usepackage{subcaption}
\usepackage{caption}

\newtheorem{corr}{Corollary}
\newtheorem{deff}{Definition}

\newtheorem{lemm}{Lemma}

\newtheorem{prob}{Problem}
\newtheorem{theo}{Theorem}
\newtheorem{assum}{Assumption}
\newtheorem{rem}{Remark}
\newtheorem{example}{Example}

\def\tt{\tilde \theta}

\def\L2{{\cal L}_2}

\newlength{\defbaselineskip}
\newcommand{\setlinespacing}[1]%
           {\setlength{\baselineskip}{#1 \defbaselineskip}}

\newcommand{\beq}{\begin{equation}}
\newcommand{\eeq}{\end{equation}}


\newcommand{\xdasharrow}[2][-->]{
\tikz[baseline=-\the\dimexpr\fontdimen22\textfont2\relax]{
\node[anchor=south,font=\scriptsize, inner ysep=1.5pt,outer xsep=8pt](x){#2};
\draw[shorten <=3.4pt,shorten >=3.4pt,dashed,#1](x.south west)--(x.south east);
}
}

\def\BibTeX{{\rm B\kern-.05em{\sc i\kern-.025em b}\kern-.08em
    T\kern-.1667em\lower.7ex\hbox{E}\kern-.125emX}}

\usepackage{longtable}
\usepackage{amsmath}
\usepackage{graphicx}
\usepackage{color}
\usepackage[dvips]{epsfig}

\usepackage{times}

\usepackage{graphicx, amssymb}
\usepackage{amsmath}

\smallskip

\def\BibTeX{{\rm B\kern-.05em{\sc i\kern-.025em b}\kern-.08em
    T\kern-.1667em\lower.7ex\hbox{E}\kern-.125emX}}

\title{
From Vertices to Convex Hulls: Certifying Set-Wise Compatibility for CBF Constraints 
}

\author{Shima Sadat Mousavi, Xiao Tan, and Aaron D. Ames %
\thanks{The authors are with the Department of Mechanical and Civil Engineering, California Institute of Technology, Pasadena, CA \texttt{\{smousavi,xiaotan,ames\}@caltech.edu}.}
\thanks{This research was supported by the Boeing Strategic University Initiative. 
}
}

\begin{document}

\maketitle

\thispagestyle{empty} 

\begin{abstract}

This paper develops certificates that propagate compatibility of multiple control barrier function (CBF) constraints from sampled vertices to their convex hull. Under mild concavity and affinity assumptions, we present three sufficient feasibility conditions {under which feasible
inputs over the convex hull can be obtained per coordinate,
with a common input, or via convex blending.} 
We also describe the associated computational methods, based on interval intersections or an offline linear program (LP). {Beyond certifying compatibility, we  give conditions under which the quadratic-program (QP) safety filter is affine in the state. This enables explicit implementations via convex combinations of vertex-feasible inputs.} 
 Case studies illustrate the results.
\end{abstract}


\section{Introduction}

Optimization-based safety filters built on Control Barrier Functions (CBFs) are widely used to  enforce forward invariance while maintaining performance 
\cite{ames2019control}. In many applications, systems must satisfy multiple state constraints under input bounds. This raises the core question of \emph{compatibility}: does a control input exist that satisfies all constraints at a given state? {Pointwise compatibility is easy to check via a convex program (e.g., a QP with hard CBF and input constraints). Region-wide certification is harder and typically requires dense sampling.}
This paper develops lightweight, set-wise certificates that propagate compatibility from finitely many sampled states (“vertices”) to all states in their convex hull.

Existing CBF--QP methods are fundamentally \emph{pointwise}: compatibility is checked at the current state and re-evaluated at the next. Canonical formulations appear in \cite{ames2019control}, and higher--relative-degree extensions retain this pointwise nature \cite{nguyen2016exponential,breeden2021high}. Recent results give sufficient conditions for pointwise feasibility and regularity: \cite{xiao2022sufficient} introduces an auxiliary feasibility constraint that preserves and enlarges the feasible set; \cite{isaly2024feasibility} provides feasibility and continuity conditions for multiple CBFs using tangent-cone and polynomial-verification tools. 

For multiple CBF constraints under input bounds, \cite{spiller2025feasibility} derives pointwise feasibility and robustness margins for bounded-input second-order systems. 
 Sampling-based methods refine state grids using Lipschitz bounds \cite{tan2022compatibility}. Closed-form QP analyses give compatibility conditions for box-constrained MIMO systems \cite{cohen2025compatibility}. SOS- and SDP-based approaches further certify feasibility for polynomial systems 
\cite{schneeberger2023sos,ClarkCDC2021, wang2023safety},  
though often with high computational cost. Overall, these methods enhance pointwise feasibility guarantees but provide limited lightweight \emph{set-wise} certification.

By contrast, we pursue \emph{lightweight, region-wide guarantees}. Using tools from convex analysis \cite{boyd2004convex}, 
we develop \emph{Compatibility Propagation Certificates (CPCs)} that extend vertex-level feasibility of CBF constraints to all states in their convex hull. Each certificate provides a structural condition under which feasibility \emph{propagates} through the region, together with an offline computational check to certify it. Specifically,

\begin{itemize}
 \item \textit{CPC--Interval.}
For box-bounded inputs, this certificate exploits 
{sign coherence of the control effectiveness matrix in the CBF condition}:
intersecting one-dimensional feasible intervals across vertices yields a region-wide {feasible} input box (Thm.~\ref{thm:vertex-sign-box}). Its special case, the \emph{Endpoint Rule}, applies when each input channel has a uniform effect, requiring no computation and giving an immediate feasibility test (Cor.~\ref{cor:endpoint-rule}).

    \item \textit{ \textcolor{black}{CPC--Common}}.
    This certificate ensures that a common, state-independent control input is feasible at all vertices, 
    and hence throughout their convex hull (Thm.~\ref{thm:single-input-propagation}). 
    It is certified by solving a single \textcolor{black} {Linear Program (LP)}, 
    yielding a precomputed {feasible} input for the entire region.

\item \textit{CPC--Blend.} This certificate ensures that convex averages of vertex inputs remain feasible for convex combinations of states under a pairwise condition (Thm.~\ref{thm:pairwise-mono-compat-fixed}). {When the control effectiveness matrix in the CBF constraints is state-independent}—as in LTI systems—this condition holds automatically, making the convex blend of vertex inputs a region-wide feasible input (Cor.~\ref{cor:pairwise-mono-constant}).

\end{itemize}

The proposed CPCs complement standard online CBF--QP methods by providing offline, region-wide feasibility guarantees over convex hulls. This enables explicit controllers without online optimization. The CPCs form a hierarchy of \emph{sufficient} conditions that balance simplicity and conservatism. Unlike SOS or grid-refinement approaches, which scale poorly with dimension~\cite{boyd2004convex,tan2022compatibility}, CPCs require only a one-time offline check using interval intersections or small LPs.

{We also show that, under mild structural conditions, the safety filter admits an \emph{explicit affine feedback form}. The control law can then be recovered from precomputed optimal {inputs} at vertices, eliminating the need for online optimization.}

\section{Preliminaries}

\textit{Notation.} 
{We use $\mathbb{R}$ and $\mathbb{R}_{\ge0}$ for the set of real and nonnegative real numbers, and $\mathbb{N}$ for the set of positive integers.} 
{The spaces $\mathbb{R}^n$ and $\mathbb{R}^{n\times m}$ denote the sets of $n$-dimensional real vectors and $n\times m$ real matrices. 
For $h:\mathbb{R}^n\!\to\!\mathbb{R}$ of class $C^1$ and $f:\mathbb{R}^n\!\to\!\mathbb{R}^n$, the gradient is $\nabla h(x)\in\mathbb{R}^{1\times n}$, and we define 
$L_f h(x)=\nabla h(x)^\top f(x)\in \mathbb{R}$. For $g:\mathbb{R}^n\to\mathbb{R}^{n\times m}$, 
$L_g h(x)=\nabla h(x)^\top g(x)\in\mathbb{R}^{1\times m}$.}
For any $N\in\mathbb{N}$, the simplex is $\Delta_N=\{\lambda\in\mathbb{R}_{\ge0}^N:\sum_j\lambda_j=1\}$, and the convex hull of points $x^1,\dots,x^N$ is
\(
\mathrm{co}\{x^1,\dots,x^N\}=\Big\{\sum_j \lambda_j x^j:\lambda\in\Delta_N\Big\}.
\)
{Vector inequalities are interpreted componentwise; for a symmetric matrix $A$, $A\succ0$ denotes positive definiteness.}
For a matrix $A\in\mathbb{R}^{n\times m}$, the entry in row $i$, column $j$ is $A_{ij}$ and the $k$th column is $A_{:,k}$. 
{The vector $\mathbf{1}_p$ denotes the $p$-dimensional all-ones vector.}
The notation $\mathcal{O}(\cdot)$ indicates asymptotic computational complexity.

\smallskip

Consider the control-affine system
\begin{equation}
\dot x = f(x) + g(x)u, 
\qquad x \in \mathcal{X}\subseteq \mathbb{R}^n,\; u \in \mathcal{U}\subseteq \mathbb{R}^m,
\label{eq:system}
\end{equation}
where $f:\mathbb{R}^n\!\to\!\mathbb{R}^n$ and $g:\mathbb{R}^n\!\to\!\mathbb{R}^{n\times m}$ are locally Lipschitz, 
and {the  input set $\mathcal{U}$ is nonempty, compact, and convex.} 
Let $\{h_i:\mathcal{X}\!\to\!\mathbb{R}\}_{i=1}^p$ be  continuously differentiable ($C^1$) functions. 
Define the superlevel sets
\(
\mathcal C_i \coloneqq \{x\in\mathcal{X}:\ h_i(x)\ge 0\},
\)
and the safe set
\(
\mathcal C \coloneqq \bigcap_{i=1}^p \mathcal C_i.
\)

\smallskip
\begin{deff}\label{def:CBF1}
 A  $C^1$ function $h_i:\mathcal X\to\mathbb R$ is a CBF for \eqref{eq:system} if there exists an extended class-$\mathcal K_\infty$ function $\alpha_i$ where
\[
\forall\,x\in\mathcal X\ \ \exists\,u\in\mathcal U:\quad
L_f h_i(x)+L_g h_i(x)\,u \ \ge\ -\alpha_i\!\big(h_i(x)\big).
\]
\end{deff}

{In Def.~\ref{def:CBF1}, we use the constrained–input version of CBFs with $u\in\mathcal U$ (cf.~\cite{ames2019control}).} 
From \cite{ames2019control}, any locally Lipschitz controller $u: x\mapsto u(x)$ that satisfies the above inequality pointwise guarantees that $x(t)\in \mathcal{C}_i$ for all $t$.

For each $i\in\{1,\dots,p\}$, let
\(
\Psi_i^\top(x)\coloneqq L_g h_i(x)\in\mathbb R^{1\times m},\;
\delta_i(x)\coloneqq L_f h_i(x)+\alpha_i\!\big(h_i(x)\big)\in\mathbb R,
\)
so the $i$th CBF inequality is $\Psi_i^\top(x)u+\delta_i(x)\ge 0$.
Stacking the $p$ rows yields
\begin{equation}
\label{eq:stacked_constraints}
\Psi(x)\,u+\delta(x)\ \ge\ 0,
\end{equation}
with $\Psi(x)\coloneqq[\Psi_1(x),\dots,\Psi_p(x)]^\top\in\mathbb R^{p\times m}$ and
$\delta(x)\coloneqq[\delta_1(x),\dots,\delta_p(x)]^\top\in\mathbb R^p$.

In order to enforce safety while maintaining performance, a \emph{safety filter}  maps each state $x\in\mathcal X\subset\mathbb R^n$ and a desired input $u_{\mathrm{des}}(x)\in \mathbb{R}^m$ to an input that keeps the system safe.  A common formulation for the CBF-based safety filter design is  
\begin{equation}
\label{eq:qp-safety}
\begin{aligned}
u^\star(x)=\arg\min_{u\in\mathcal U}\ & \tfrac12\|u-u_{\mathrm{des}}(x)\|^2\\
\text{s.t.}\ & \Psi(x)\,u+\delta(x)\ \ge\ 0.
\end{aligned}
\end{equation}

We say the CBF inequalities in \eqref{eq:stacked_constraints} are \emph{compatible at $x$} if there exists $u\in\mathcal U$ with  $\Psi(x)u+\delta(x)\ge 0$.
A natural question is how to certify compatibility over a region $\mathcal{X}$.  {Def.~\ref{def:CBF1} is per-constraint and does not ensure compatibility of~\eqref{eq:stacked_constraints} at $x$. Compatibility requires a single input $u\in\mathcal U$ that satisfies all inequalities at $x$ simultaneously.} In this work, we  propose several easy-to-check conditions using convex analysis. We start with two definitions.

\begin{deff}
A function $\varphi:\mathbb{R}^n \to \mathbb{R}$ is \emph{concave} on a convex set $\mathcal{X}$ if
\(
\varphi\!\left(\sum_{j=1}^N \lambda_j x^j\right) \ge \sum_{j=1}^N \lambda_j\,\varphi(x^j)
\)
for all $x^j\in\mathcal{X}$ and all $\lambda\in\Delta_N$.
Concavity is \emph{componentwise} for vector-valued functions. If the inequality is reversed, $\varphi$ is \emph{convex}  on $\mathcal{X}$ \cite{boyd2004convex}. 
\end{deff}

\begin{deff}
\label{def:compat-U}
A finite set $\{x^1,\dots,x^N\}$ has \emph{vertex compatibility} if for each $j$ there exists a  {\emph{vertex-feasible input}} $u^j\in\mathcal U$ such that $\Psi(x^j)u^j+\delta(x^j)\ge 0$.  {We call an input $u\in\mathcal U$ \emph{feasible} at $x$ if it satisfies the stacked CBF
inequalities \eqref{eq:stacked_constraints}.}
\end{deff}

We seek to address the following problem.

\begin{prob}
\label{prob:compat-hull-U}
Given $\Psi(x), \delta(x)$ from \eqref{eq:stacked_constraints}, a convex input set $\mathcal U$, and vertices $\{x^1,\dots,x^N\}$ with vertex compatibility, provide conditions ensuring compatibility on the hull
$H\coloneqq\operatorname{co}\{x^1,\dots,x^N\}$; i.e., for every $x\in H$ there exists $u(x)\in\mathcal U$  that satisfies \eqref{eq:stacked_constraints}.
\end{prob}

\section{Compatibility over Convex Hulls}

This section gives conditions under which compatibility at finitely many \emph{vertices} extends to the full convex hull. Each result targets a different structural case.

\begin{assum}
\label{assum:standing-rev}
Throughout, unless stated otherwise:
\begin{enumerate}
\item[(A1)] \textit{Vertex compatibility (with input constraints):} for each $j=1,\dots,N$ there exists a vertex input  $u^j\in\mathcal U$ such that
$\Psi(x^j)u^j+\delta(x^j)\ge 0$.
\item[(A2)] \textit{Concave drift on $H$:} The map $\delta:H\to\mathbb R^p$ is  concave.
\item[(A3)] \textit{Columnwise curvature alignment:}
for each column $k\in\{1,\dots,m\}$, the column map $\Psi_{:,k}(\cdot):H\to\mathbb R^p$ is concave, convex, or affine on $H$.
\end{enumerate}
\end{assum}

Under (A3), define for each input coordinate $k$: 
\begin{equation}
\label{eq:Spsi_k}
\mathcal S_k(\Psi)\triangleq
\begin{cases}
[0,\infty), & \text{if }\Psi_{:,k}(\cdot)\text{ is concave on }H,\\[2pt]
(-\infty,0], & \text{if }\Psi_{:,k}(\cdot)\text{ is convex on }H,\\[2pt]
\mathbb R, & \text{if }\Psi_{:,k}(\cdot)\text{ is affine on }H.
\end{cases}
\end{equation}
The \emph{sign-aligned cone} is then
\begin{equation}
\label{eq:Spsi}
\mathcal S(\Psi)\ \triangleq\ \{\,u\in\mathbb R^m:\ u_k\in\mathcal S_k(\Psi)\ \text{for all }k\,\}.
\end{equation}

\begin{rem}
If $\Psi$ is concave (resp.\ convex) on $H$,  $\mathcal S(\Psi)=\{u\ge 0\}$ (resp.\ $\{u\le 0\}$).
If $\Psi$ is affine or constant on $H$, then  $\mathcal S(\Psi)=\mathbb R^m$. 
{In particular, for an LTI plant $\dot{x}=A x + B u$ with affine CBFs $h_i(x)=a_i^{\top}x+b_i$, 
one has $L_g h_i(x)=a_i^{\top}B$ that is constant on $H$. 
Hence $\Psi$ is constant, and thus $\mathcal{S}(\Psi)=\mathbb{R}^m$.} 
\end{rem}

\begin{lemm}
\label{lem:phi-concave}
Under \textup{(A3)}, for any $u\in\mathcal S(\Psi)$ the map $\phi(x)\triangleq\Psi(x)u$ is concave on $H$.
\end{lemm}
\begin{proof}
 Each component is a sum of column maps weighted by $u_k$, where concave columns carry nonnegative weights, convex columns carry nonpositive weights, and affine columns carry arbitrary weights; (nonnegative)$\times$concave $+$ (nonpositive)$\times$convex $+$ affine is concave on $H$.
\end{proof}

\begin{rem}
{Our results  require only a concave \emph{lower bound} on $\Psi(x)u+\delta(x)$ over the hull. 
 {Under (A3), $u\in \mathcal{S}(\Psi)$  ensures that $x\mapsto\Psi(x)u$ is concave. Thus, (A2) can be replaced by a concave lower bound $\underline{\delta}(\cdot)$ with $\underline{\delta}(x)\le\delta(x)$ for all $x\in H$.} }
\end{rem}

The following example shows that vertex compatibility alone does not guarantee hull compatibility. This motivates the CPC conditions developed next.

\begin{example}\label{ex:infeasibility} Consider the one-dimensional case
\[
\Psi(x)=\left[\begin{smallmatrix}
-(x-4)^2\\[1pt]
1 \end{smallmatrix}\right],\qquad
\delta(x)= \left[\begin{smallmatrix}
-x+10\\[1pt]
-x \end{smallmatrix}\right],\qquad
\mathcal U=[0,10],
\]
with vertices $x^1=0$, $x^2=3$, so $H=[0,3]$.  
Assump.~\ref{assum:standing-rev} holds with {vertex-feasible inputs} $u^1=0.5$ and $u^2=3$.  
However, the condition $\Psi(1.5)u+\delta(1.5)\ge0$ is infeasible, showing that compatibility fails inside $H$.
\end{example}

\subsection{Certificate I: {CPC--Interval}}

We begin with the fastest, optimization-free certificate for box inputs: the \emph{CPC--Interval}. It constructs one-dimensional feasible \emph{intervals} for each input coordinate
from the vertex constraints.
Intersecting these intervals yields an input box that is feasible for all states
in the convex hull.

Assume a set of box inputs $\mathcal U=\{u\in\mathbb R^m: u_{k,\min}\le u_k\le u_{k,\max}\}$ and known
 {vertex-feasible inputs} $\{u^j\}_{j=1}^N\subset\mathcal U$.
Fix an input coordinate $k$.
At each vertex $x^j$, the column $\Psi_{:,k}(x^j)$ captures how $u_k$ affects the $p$ CBF constraints.
Assume this column is \emph{sign-coherent} across the \emph{vertices}: it is entrywise nonnegative or entrywise nonpositive at every vertex.

Define
\[
\mathcal P_k=\{j:\Psi_{:,k}(x^j)\ge0\},\qquad 
\mathcal N_k=\{j:\Psi_{:,k}(x^j)\le0\},
\]
which collect, respectively, the vertices where the $k$-th input has a positive or negative influence on the CBF constraints.  Using the {vertex-feasible inputs} $u^j$, set
\[
L_k=\max_{j\in\mathcal P_k}(u^j)_k,\qquad
U_k=\min_{j\in\mathcal N_k}(u^j)_k,
\]
with $L_k=-\infty$ if $\mathcal P_k=\emptyset$ and $U_k=+\infty$ if $\mathcal N_k=\emptyset$. 
Here, $L_k$ is the largest input value valid for all ``positive'' vertices in $\mathcal P_k$, and
$U_k$ is the smallest input value valid for all ``negative'' vertices in $\mathcal N_k$.
The resulting feasible interval of the $k$-th input that satisfies both \eqref{eq:stacked_constraints} and actuator limits is then  
\(
I_k=[\,\max\{L_k,u_{k,\min}\},\ \min\{U_k,u_{k,\max}\}\,].
\)

{To ensure that the selected inputs respect the curvature of $\Psi$, 
each coordinate interval $I_k$ is refined by intersecting it with the corresponding 
sign-aligned set $\mathcal S_k(\Psi)$ defined in \eqref{eq:Spsi_k}:
\begin{equation}\label{eq:intersec}
\tilde I_k = I_k \cap \mathcal S_k(\Psi).
\end{equation}
This step narrows each input’s feasible interval to the values that match the curvature direction of the CBF constraints. If any $\tilde I_k$ is empty, the CPC--Interval test fails; otherwise, any $u$ with $u_k\in\tilde I_k$ for all $k$ guarantees compatibility over the  hull.}

\begin{theo}
\label{thm:vertex-sign-box}
If the intervals $\tilde I_k$ are nonempty for all $k$, then any $u$ with $u_k\in\tilde I_k$
satisfies  \eqref{eq:stacked_constraints} for all $x\in H$.
\end{theo}

\begin{proof}
Fix any $u$ with $u_k\in\tilde I_k$ for all $k$, and let $x=\sum_{j=1}^N\lambda_j x^j\in H$ with $\lambda\in\Delta_N$.
By curvature alignment, $x\mapsto\Psi(x)u$ is concave (Lem.~\ref{lem:phi-concave}); $x\mapsto\delta(x)$ is concave as well, so
\(
\Psi(x)u+\delta(x)\ \ge\ \sum_{j=1}^N \lambda_j\big(\Psi(x^j)u+\delta(x^j)\big).
\)
It suffices to show $\Psi(x^j)u+\delta(x^j)\ge 0$ for each vertex $x^j$.
For any $j$ and $k$: if $j\in\mathcal P_k$, then $u_k\ge L_k\ge (u^j)_k$ and $\Psi_{:,k}(x^j)\ge 0$ give
$\Psi_{:,k}(x^j)u_k\ge \Psi_{:,k}(x^j)(u^j)_k$; if $j\in\mathcal N_k$, then $u_k\le U_k\le (u^j)_k$ and
$\Psi_{:,k}(x^j)\le 0$ give the same inequality. Summing over $k$ yields
$\Psi(x^j)u\ge \Psi(x^j)u^j$, and feasibility of $u^j$ implies
$\Psi(x^j)u+\delta(x^j)\ge 0$. Thus, by concavity, $\Psi(x)\,u+\delta(x)\ \ge\ 0$, for all $x\in H$.
\end{proof}

Thm.~\ref{thm:vertex-sign-box} propagates vertex compatibility to the whole hull 
when all coordinate intervals $\tilde I_k$ are nonempty. 
When  {vertex-feasible inputs} $\{u^j\}$ are unavailable, 
a simple alternative is to construct them directly from the column signs of $\Psi$. 
This yields the  \emph{Endpoint Rule}.

\begin{corr}
\label{cor:endpoint-rule}
Assume: (i) each column of $\Psi$ has a  uniform sign over $H$; 
(ii) $\mathcal U=\prod_{k=1}^m [u_{k,\min},u_{k,\max}]$; and 
(iii) $E_k:=[u_{k,\min},u_{k,\max}]\cap \mathcal S_k(\Psi)\neq\emptyset$ for all $k$.
Define
\[
u_k \;=\;
\begin{cases}
\sup E_k, & \Psi_{:,k}(x)\ge 0 \ \ \forall x\in H,\\[2pt]
\inf E_k, & \Psi_{:,k}(x)\le 0 \ \ \forall x\in H,
\end{cases}
\quad k=1,\dots,m,
\]
and set $u=(u_1,\dots,u_m)$. Then $u\in \mathcal U\cap \mathcal S(\Psi)$, and $u$ satisfies 
\eqref{eq:stacked_constraints} for all $x\in H$.
If any $E_k=\emptyset$, the test is inconclusive.
\end{corr}

\begin{proof}
Under (A2)–(A3) and uniform sign over $H$, Thm.~\ref{thm:vertex-sign-box}
gives per–coordinate feasible intervals 
$I_k=[\max\{L_k,u_{k,\min}\},\,\min\{U_k,u_{k,\max}\}]$ with 
$U_k=+\infty$ if $\Psi_{:,k}\!\ge0$ on $H$ and $L_k=-\infty$ if $\Psi_{:,k}\!\le0$. 
Let $\tilde I_k=I_k\cap\mathcal S_k(\Psi)$ and 
$E_k=[u_{k,\min},u_{k,\max}]\cap\mathcal S_k(\Psi)\neq\emptyset$.
If $\Psi_{:,k}\!\ge0$ on $H$, choosing $u_k=\sup E_k$ ensures 
$u_k\in E_k\subseteq \tilde{I}_k$, hence $u_k\in\tilde I_k$; 
if $\Psi_{:,k}\!\le0$, the same holds with $u_k=\inf E_k$.
Thus $u=(u_1,\dots,u_m)\in\prod_k\tilde I_k\subseteq\mathcal U\cap\mathcal S(\Psi)$,
and by Thm.~\ref{thm:vertex-sign-box}, $\Psi(x)u+\delta(x)\ge0$ for all $x\in H$.
\end{proof}

The Endpoint Rule represents a conservative instance of CPC--Interval, 
requiring no optimization or {vertex-feasible inputs}. 
It is especially useful for quick compatibility checks when the column signs of $\Psi$ are sign-coherent across a region.

\subsection{Certificate II: {CPC--Common}}

The CPC--Common certificate checks whether one \emph{common} input can satisfy all vertex constraints. Rather than constructing per-vertex inputs, we solve one LP to find \(u^\star\in\mathcal U\cap\mathcal S(\Psi)\)  satisfying the CBF constraints at all vertices. If the LP is feasible, the same \(u^\star\) satisfies the  CBF constraints for every \(x\in H\)  by (A2) and Lem.~\ref{lem:phi-concave}.

\begin{theo}\label{thm:single-input-propagation}
If there exists a common $u\in\mathcal U\cap\mathcal S(\Psi)$ with $\Psi(x^j)u+\delta(x^j)\ge 0$ for all vertices $x^j$, $u$ satisfies \eqref{eq:stacked_constraints} for all $x\in H$.
\end{theo}

\begin{proof}
Let $x=\sum_j\lambda_j x^j$. By Lem.~\ref{lem:phi-concave}, $x\mapsto\Psi(x)u$ is concave; $\delta(\cdot)$ is concave by Assump. \ref{assum:standing-rev}. Hence
\(
\Psi(x)u+\delta(x)\ \ge\ \sum_j\lambda_j\big(\Psi(x^j)u+\delta(x^j)\big)\ \ge\ 0.
\)
\end{proof}

\smallskip

{To apply CPC--Common, we solve a single LP using the vertex data
$\{\Psi(x^j),\delta(x^j)\}_{j=1}^N$.
The LP searches for an input $u\in\mathcal U\cap\mathcal S(\Psi)$ that maximizes
a common vertex margin:}
\begin{equation}
\label{eq:single-input-margin-lp}
{
\begin{aligned}
u^{\star},\,t^{\star}
&= \arg\max_{u,\,t}\ t\\
\text{s.t.}\quad 
& \Psi(x^j)u+\delta(x^j)\ \ge\ t\,\mathbf{1}_p,\qquad j=1,\dots,N,\\
& u\in \mathcal U \cap \mathcal S(\Psi).
\end{aligned}}
\end{equation}
 {Let $(u^\star,t^\star)$ denote the optimizer of \eqref{eq:single-input-margin-lp}.}
If $t^\star\ge 0$, then $u^\star$ satisfies~\eqref{eq:stacked_constraints} for all $x\in H$.

\begin{rem}
LP~\eqref{eq:single-input-margin-lp} has $m{+}1$ variables {($m$ inputs)} and $pN$ vertex constraints
 {(from $p$ CBF inequalities at $N$ vertices)}.
It also includes the input–set constraints (e.g., $q$ rows if $\mathcal U=\{u:Gu\le b\}$ with $G\in \mathbb{R}^{q\times m },\ b\in \mathbb{R}^q$)
and at most $m$ sign constraints from $\mathcal S(\Psi)$.
A dense
interior-point step costs
$\mathcal O\!\big((m{+}1)^3\big) + \mathcal O\!\big((pN{+}q{+}m)(m{+}1)^2\big)$
per iteration \cite{boyd2004convex}. 
The number of iterations is $\mathcal O\!\big(\sqrt{pN{+}q{+}m}\,\log(1/\varepsilon)\big)$, where $\varepsilon$ denotes 
how close the returned solution is to the true optimum \cite{boyd2004convex}.   In practice, only a few tens of iterations are  required. The LP is solved once offline; no per-state online optimization is needed. 
\end{rem}

\subsection{Certificate III: {CPC--Blend}}

{If no common input is feasible at all vertices (i.e., LP~\eqref{eq:single-input-margin-lp} is infeasible), we use \emph{CPC--Blend}. This certificate guarantees that convex blending of vertex states and their  {vertex-feasible inputs} preserves constraint compatibility.}

\begin{theo}
\label{thm:pairwise-mono-compat-fixed}
Let $\{u^j\}\subseteq\mathcal U\cap\mathcal S(\Psi)$ satisfy $\Psi(x^j)u^j+\delta(x^j)\ge 0$ at all vertices.
If
\begin{equation}\label{eq:pairwise_monotone}
\big(\Psi(x^i)-\Psi(x^j)\big)\big(u^i-u^j\big)\ \le\ 0\qquad \text{for all }\, i,j,
\end{equation}
then for any $x=\sum_j\lambda_j x^j$ with $\lambda\in\Delta_N$, the convex blend
$u_\lambda=\sum_j\lambda_j u^j$ satisfies \eqref{eq:stacked_constraints}.
\end{theo}

\begin{proof}
Let $A^{(j)}:=\Psi(x^j)$ and $S:=\sum_j\lambda_j A^{(j)}$.
Since $\mathcal U$ is convex and $\mathcal S(\Psi)$ is a cone, $u_\lambda\in\mathcal U\cap\mathcal S(\Psi)$.
By Lem.~\ref{lem:phi-concave}, $\Psi(\cdot)u_\lambda$ is concave, so
\(
\Psi(x)u_\lambda \ \ge\ \sum_j\lambda_j A^{(j)}u_\lambda \ =\ S\,u_\lambda.
\)
A symmetric expansion gives
\[
S\,u_\lambda-\sum_j\lambda_j A^{(j)}u^j
=\tfrac12\sum_{i,j}\lambda_i\lambda_j\big(A^{(i)}-A^{(j)}\big)\big(u^j-u^i\big),
\]
that is nonnegative by \eqref{eq:pairwise_monotone}. Hence $S\,u_\lambda\ge \sum_j\lambda_j A^{(j)}u^j$.
Concavity of $\delta$ yields $\delta(x)\ge \sum_j\lambda_j\delta(x^j)$; therefore,
\(
\Psi(x)u_\lambda+\delta(x)\ \ge\ \sum_j\lambda_j\big(A^{(j)}u^j+\delta(x^j)\big)\ \ge\ 0,
\)
using vertex compatibility. This completes the proof.
\end{proof}

\begin{rem}\label{rem:lp certificate 3}
The CPC--Blend certificate requires verifying the pairwise condition~\eqref{eq:pairwise_monotone} once vertex inputs $\{u^j\}$ are available. 
These inputs can be obtained in two ways:

\smallskip
\noindent(i) \emph{Per-vertex approach:} Solve $N$ small LPs (one per vertex) to find feasible $\{u^j\}$, then apply the pairwise check~\eqref{eq:pairwise_monotone}. 
This approach is simple and parallelizable but may fail when the resulting $\{u^j\}$ violate~\eqref{eq:pairwise_monotone}. 

\smallskip
\noindent(ii) \emph{Joint LP approach:} Solve one combined LP that searches for all $\{u^j\}$ simultaneously while enforcing both vertex feasibility and pairwise condition~\eqref{eq:pairwise_monotone}:
\begin{equation}
\label{eq:joint-pairwise-lp}
\begin{aligned}
\{u^{j\star}\},\,t^{\star}
&= \arg\max_{\{u^j\},\,t}\ t\\
\text{s.t.}\quad 
& \Psi(x^j)u^j+\delta(x^j)\ \ge\ t\,\mathbf 1_p,\quad j=1,\dots,N,\\
& \big(\Psi(x^i)-\Psi(x^j)\big)\big(u^i-u^j\big)\ \le\ 0,\quad \forall\, i<j,\\
& u^j\in \mathcal U\cap\mathcal S(\Psi),\quad j=1,\dots,N.
\end{aligned}
\end{equation}
{Let $t^\star$ be its optimal value when \eqref{eq:joint-pairwise-lp} is feasible.
If $t^\star\ge 0$, then (A1) holds at all vertices with margin $t^\star$, and \eqref{eq:pairwise_monotone} holds as well.
Thm.~\ref{thm:pairwise-mono-compat-fixed} then gives compatibility over $H$.}
If \eqref{eq:joint-pairwise-lp} is infeasible, or if $t^\star<0$, the test is inconclusive.

\smallskip
\noindent {The joint LP~\eqref{eq:joint-pairwise-lp} offers a \emph{single-shot certificate}.
 {It is computationally heavier than the per-vertex approach, but it avoids failures from vertex inputs that are feasible individually yet violate \eqref{eq:pairwise_monotone}.}
In practice, the per-vertex method is preferable for large $N$ or parallel setups, while the joint LP is better suited for moderate $N$ or when a global certificate is desired.}

\smallskip
\noindent \emph{Complexity.} 
The pairwise check costs $O(N^2 p m)$ { ($p$ CBF constraints, $m$ inputs)}.
 {The per-vertex approach involves $N$ LPs, each with $m$ variables and roughly $p{+}q{+}m$ constraints.
By contrast, the joint LP has $Nm{+}1$ variables and about $pN+\binom{N}{2}p$ constraints.
Although both are polynomial in size, the latter is heavier in a single solve.} Both are tractable and executed once offline, incurring no online computational cost.
\end{rem}

As an immediate corollary of Thm~\ref{thm:pairwise-mono-compat-fixed}, if $\Psi$ is constant on $H$, then \eqref{eq:pairwise_monotone} holds trivially (the left side is $0$).

\begin{corr}\label{cor:pairwise-mono-constant}
If $\Psi(x)\equiv\bar\Psi$ on $H$, for any
$x=\sum_{j=1}^N\lambda_j x^j$ with $\lambda\in\Delta_N$, 
$u_\lambda=\sum_{j=1}^N\lambda_j u^j$ satisfies \eqref{eq:stacked_constraints}.
\end{corr}

\begin{rem}
For systems with state–independent input channels $g(x)$ (e.g., LTI or linearized dynamics $\dot x = Ax + Bu$) and affine CBFs $h_i(x)=a_i^\top x+b_i$, we have $\Psi_i(x)=a_i^\top B$, hence $\Psi$ is constant. If, in addition, the class-$\mathcal K$ functions $\alpha_i$ are concave (e.g., $\alpha_i(s)=\kappa_i s$ with $\kappa_i\ge 0$), then
\(
\delta_i(x)=a_i^\top A x + \alpha_i(a_i^\top x + b_i)
\)
is concave on $H$.  
Thus, Assump.~\ref{assum:standing-rev} holds with constant $\Psi$, and Cor.~\ref{cor:pairwise-mono-constant} applies directly—no LP is required.  
This covers cases such as {Adaptive Cruise Control} and lane keeping \cite{Rajamani2012}, velocity–controlled mobile robots \cite{LynchPark2017}, and multirotors near hover under linearization \cite{MahonyKumarCorke2012}.
\end{rem}

{\begin{example}\label{ex:condition}
Revisiting Ex.~\ref{ex:infeasibility}, we note three failures:  
(i) $\Psi(\cdot)$ is not sign-coherent, (ii) the vertex-feasible inputs share no common value, and (iii) \eqref{eq:pairwise_monotone} is violated.
\end{example}}

\subsection{Affine Interpolation of the  Safety Filter}

{In many applications, the safety filter is evaluated repeatedly along trajectories in $H$.
If $\Psi(x)$, $\delta(x)$, and $u_{\mathrm{des}}(x)$ are affine on $H$, in general, the safety-filter optimizer
$u^\star(\cdot)$ is piecewise affine.
We next give conditions under which $u^\star(\cdot)$ on $H$ is equal to the affine
interpolation of the optimizer values $u^\star(x^j)$ at the vertices.}

Let the input set be a polytope
\[
\mathcal U=\{u\in\mathbb R^m:\ G u\le b\},\qquad 
G\in\mathbb R^{q\times m},\ b\in\mathbb R^q
\]
(e.g., a box with \(G=\begin{bmatrix}I\\ -I\end{bmatrix}\),
\(b=\begin{bmatrix}u_{\max}\\ -u_{\min}\end{bmatrix}\)).
For a given state \(x\), let \(\mathcal A(x)\subseteq\{1,\dots,p\}\) denote the
CBF rows active at the optimal control \(u^\star(x)\)  in \eqref{eq:qp-safety} and
\(\mathcal B(x)\subseteq\{1,\dots,q\}\) the active input bounds:
\[
\Psi_i(x)u^\star(x)+\delta_i(x)=0 \ (i\in\mathcal A),\qquad
G_\ell u^\star(x)=b_\ell \ (\ell\in\mathcal B).
\]
Let \(\Psi_{\mathcal A}(x)\), \(\delta_{\mathcal A}(x)\) collect the active CBF rows,
and \(G_{\mathcal B}, b_{\mathcal B}\) the active input rows.
A \emph{critical region} \(\mathcal C\) is a set of states where  
(i) the active sets \((\mathcal A,\mathcal B)\) remain constant,  
(ii) the Linear Independence Constraint Qualification (LICQ) holds
(i.e., the rows of \(\big[\Psi_{\mathcal A}(x);\ G_{\mathcal B}\big]\) are linearly independent), and  
(iii) strict complementarity holds (i.e., all active multipliers are positive) (see, e.g., \cite{tondel2003algorithm}).

\begin{assum}
\label{assum:affine-mpqp}
On $H=\mathrm{co}\{x^1,\dots,x^N\}$, the problem data are constant or affine in $x$:
(i) $\Psi(x)\equiv\bar\Psi\in\mathbb{R}^{p\times m}$;
(ii) $\delta(x)=\bar{\delta}_0+\bar{\Delta}x$, and {$u_{\mathrm{des}}(x)$ is affine in $x$, i.e., 
$u_{\mathrm{des}}(x)=\bar{u}_0+\bar{U}x$}; 
(iii) $\mathcal U=\{u:Gu\le b\}$.
{Here, barred symbols (\,$\bar\Psi,\bar{\delta}_0,\bar{\Delta},\bar{u}_0,\bar{U}$\,) denote constant matrices or vectors of appropriate dimensions defining the constant/affine maps.}
\end{assum}

\begin{theo}
\label{thm:affine-interp}
Under Assump.~\ref{assum:affine-mpqp}, let $\mathcal R\subseteq H$ be a convex critical region. Then,  $u^\star(\cdot)$ in \eqref{eq:qp-safety} is affine on $\mathcal R$; in particular, for any $x=\sum_j \lambda_j x^j\in\mathcal R$ with $\{x^j\}_{j=1}^N\subset \mathcal R$,  $\lambda\in\Delta_N$,
\begin{equation}\label{eq:interpol}
    u^\star(x)=\sum_j \lambda_j\,u^\star(x^j).
\end{equation}
\end{theo}

\begin{proof}
Under Assump.~\ref{assum:affine-mpqp} {(so $\Psi(x)\equiv\bar\Psi$ on $H$)}, the
Karush--Kuhn--Tucker  conditions of \eqref{eq:qp-safety}  on $\mathcal R$ are
\(
 u-u_{\mathrm{des}}(x)+\bar\Psi_{\mathcal A}^{\!\top}\lambda_{\mathcal A}+G_{\mathcal B}^{\!\top}\nu_{\mathcal B}=0,\;
 \bar\Psi_{\mathcal A}\,u+\delta_{\mathcal A}(x)=0,\;G_{\mathcal B}u=b_{\mathcal B},
\)
with multipliers \(\lambda_{\mathcal A},\nu_{\mathcal B}>0\).
Equivalently,
\[
\underbrace{\begin{bmatrix}
I & \bar\Psi_{\mathcal A}^{\!\top} & G_{\mathcal B}^{\!\top}\\
\bar\Psi_{\mathcal A} & 0 & 0\\
G_{\mathcal B} & 0 & 0
\end{bmatrix}}_{K\ \text{constant on }\mathcal R}
\!\begin{bmatrix}u\\ \lambda_{\mathcal A}\\ \nu_{\mathcal B}\end{bmatrix}
=
\begin{bmatrix}
\,u_{\mathrm{des}}(x)\\ -\delta_{\mathcal A}(x)\\ b_{\mathcal B}
\end{bmatrix}.
\]
The right-hand side is affine in $x$. With LICQ, $K$ is nonsingular, so
$(u,\lambda_{\mathcal A},\nu_{\mathcal B})=K^{-1}(\cdot)$ is affine in $x$.
Therefore, $u^\star(\cdot)$ is affine on $\mathcal R$, and \eqref{eq:interpol} follows.
\end{proof}

\begin{corr}
\label{cor:hull-interp}
{
Under Assump.~\ref{assum:affine-mpqp}, 
if (i) all vertices $\{x^j\}_{j=1}^N\subset H$ solve \eqref{eq:qp-safety} with a common active set $(\mathcal A,\mathcal B)$ satisfying LICQ, 
and (ii) every inactive constraint remains strictly inactive on $H$, i.e.,
\(
\Psi_i(x)u^\star(x)+\delta_i(x)>0\ \ (i\notin\mathcal A),\;
b_\ell-G_\ell u^\star(x)>0\ \ (\ell\notin\mathcal B),
\)
then $H$ is a single critical region and the optimal control $u^\star(\cdot)$ is affine on $H$.
}
\end{corr}

\begin{proof}
Since the constraint residuals 
$\Psi_i(x)u^\star(x)+\delta_i(x)$ and $b_\ell-G_\ell u^\star(x)$ 
are affine in $x$, strict positivity at the vertices implies strict positivity 
throughout $H$, ensuring that no constraint changes its activity within $H$. Hence the active set $(\mathcal A,\mathcal B)$ remains constant, so
$H$ constitutes one critical region. The results thus follow  from Thm.~\ref{thm:affine-interp}. 
\end{proof}

Under Assump.~\ref{assum:affine-mpqp}, the set $H$ admits a finite
polyhedral partition into \emph{critical regions}. On each region, 
$u^\star(\cdot)$ is affine (Thm.~\ref{thm:affine-interp}).  
LICQ and strict complementarity ensure continuity across shared boundaries. 
Hence $u^\star$, being piecewise affine on compact $H$, is globally Lipschitz.

\section{Numerical  Examples}

This section illustrates how the certificates are applied to check compatibility of CBF constraints over convex hulls.
 {Table~\ref{tab:workflow} compares the certificates and indicates when each applies. As a rule of thumb,
CPC--Interval applies with sign coherence, while CPC--Common and CPC--Blend cover more general cases. Note that the certificates are only sufficient; if all fail, compatibility remains inconclusive.}

\begin{table}[t]
\centering
\renewcommand{\arraystretch}{1.2}
\caption{{Comparison and applicability of CPCs.}}
\label{tab:workflow}

{
\begin{tabular}{p{0.38\linewidth}p{0.55\linewidth}}
\toprule
\textbf{Prerequisite} & \textbf{Certificate} \\
\midrule
Box inputs, w.o. vertex inputs 
& Endpoint Rule (Cor.~\ref{cor:endpoint-rule}); no optimization. \\

Box inputs, w. vertex inputs  
& CPC--Interval (Thm.~\ref{thm:vertex-sign-box}); fast interval test. \\

Common feasible input exists 
& CPC--Common (Thm.~\ref{thm:single-input-propagation}); one offline LP. \\

No common feasible input; Pairwise condition \eqref{eq:pairwise_monotone} holds 
& CPC--Blend (Thm.~\ref{thm:pairwise-mono-compat-fixed} or Cor.~\ref{cor:pairwise-mono-constant} when $\Psi$ is constant); One offline LP w. larger size. \\
\bottomrule
\end{tabular}
}
\end{table}

Consider a three-room temperature–regulation model \cite{girard2015safety}.  
Let $x_i$ be the temperature of room $i$ ($i=1,2,3$), influenced by adjacent rooms, the
environment, and a local heater:
\[
\dot{x}_{i}
= a(x_{i+1}+x_{i-1}-2x_i)
+ b(t_e - x_i)
+ c(t_{h,i}-x_i)\,u_i,
\]
with cyclic indexing $x_0=x_3$, $x_4=x_1$.  
Here $t_e$ and $t_{h,i}$ are the environment and heater temperatures,  
$a,b,c$ are conduction coefficients, and $u_i\in[0,1]$ is the heater valve control.

\textit{Case 1.} Each room must stay above $25^\circ\!$C, so the safe set is  
\(
\mathcal C=\{(x_1,x_2,x_3):\ x_i-25\ge0,\ i=1,2,3\}.
\)
With parameters $(t_e,a,b,c)=(-1,\,0.05,\,0.06,\,0.08)$ and $t_{h,i}=50$,  
the  CBFs are $h_i(x)=x_i-25$.   With $\alpha(s)=s$, the stacked data in~\eqref{eq:stacked_constraints} is
\[
\Psi=4I-0.08\,\mathrm{Diag}(x),
\quad
\delta=\left[\begin{smallmatrix}
0.84x_1+0.05x_2+0.05x_3-25.06\\[1pt]
0.05x_1+0.84x_2+0.05x_3-25.06\\[1pt]
0.05x_1+0.05x_2+0.84x_3-25.06
\end{smallmatrix}\right].
\]
where $\mathrm{Diag}(x)$ denotes the diagonal matrix with diagonal entries $(x_1,x_2,x_3)$.
Over the polytope $H=\{(x_1,x_2,x_3):25\le x_i\le 30\}$, Assump.~\ref{assum:standing-rev} holds, and each column of $\Psi$ has a uniform sign.  
A feasible vertex input is $u=(0.78,0.78,0.78)$ at $x=(25,25,25)$ and $u=(0,0,0)$ at all other vertices.  
From these, \eqref{eq:intersec} yields the  intervals $\tilde I_k=[0.72,1]$ for $k=1,2,3$.  
By Thm.~\ref{thm:vertex-sign-box}, any $u$ in these intervals satisfies the stacked CBF constraints on  $H$, so the three CBFs are compatible over $H$. 

\textit{Case 2.} Now require each room temperature to lie in $[25,30]$, so the safe set is  
\(
\mathcal C=\{(x_1,x_2,x_3): x_i\!-\!25\ge0,\ 30\!-\!x_i\ge0,\ i=1,2,3\}.
\)
Using the same parameters as in Case~1, we compute $\Psi(x)$ and $\delta(x)$ as in~\eqref{eq:case2_num}.  
Here the hull $H$ coincides with $\mathcal C$, but the columns of $\Psi$ change sign over $H$, so Thm.~\ref{thm:vertex-sign-box} is not applicable.  
Instead, solving the LP in~\eqref{eq:single-input-margin-lp} once yields a  common {feasible} input $u=(0.78,0.78,0.78)$.  
By Thm.~\ref{thm:single-input-propagation}, this input satisfies all stacked CBF constraints on $H$, verifying compatibility of the six CBFs.

We illustrate this by simulating the room–temperature system from $10$ random initial states in $\mathcal C$ (Fig.~\ref{fig:case2_result}).  
Using the nominal controller  
\(
u_i = 0.05(x_{i+1}+x_{i-1}-2x_i)+0.05(25-x_i),
\) 
all trajectories eventually leave $\mathcal C$.  
In contrast, applying the constant input $u=(0.78,0.78,0.78)$ keeps every trajectory within $\mathcal C$.  
{Finally, the CBF-QP safety filter is always feasible and all trajectories are safe. The system moves toward the coolest admissible corner.}

\begin{figure*}
\vspace{2mm}
    \begin{equation} \label{eq:case2_num}
    \Psi =\left[\begin{smallmatrix}4.0 - 0.08 x_{1} & 0 & 0\\0.08 x_{1} - 4.0 & 0 & 0\\0 & 4.0 - 0.08 x_{2} & 0\\0 & 0.08 x_{2} - 4.0 & 0\\0 & 0 & 4.0 - 0.08 x_{3}\\0 & 0 & 0.08 x_{3} - 4.0\end{smallmatrix}\right]
,\quad \delta = \left[\begin{smallmatrix}0.84 x_{1} + 0.05 x_{2} + 0.05 x_{3} - 25.06\\- 0.84 x_{1} - 0.05 x_{2} - 0.05 x_{3} + 30.06\\0.05 x_{1} + 0.84 x_{2} + 0.05 x_{3} - 25.06\\- 0.05 x_{1} - 0.84 x_{2} - 0.05 x_{3} + 30.06\\0.05 x_{1} + 0.05 x_{2} + 0.84 x_{3} - 25.06\\- 0.05 x_{1} - 0.05 x_{2} - 0.84 x_{3} + 30.06\end{smallmatrix}\right]
\end{equation}
\end{figure*}

\begin{figure*}[ht]
	\centering
	\begin{subfigure}[t]{0.30\linewidth}		\includegraphics[width=.83\linewidth]{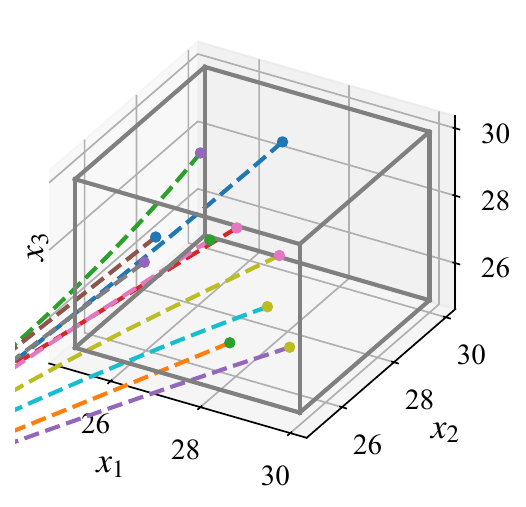}
		\caption{  Nominal controller } 
	\end{subfigure}
	\begin{subfigure}[t]{0.30\linewidth}
		\centering\includegraphics[width=.83\linewidth]{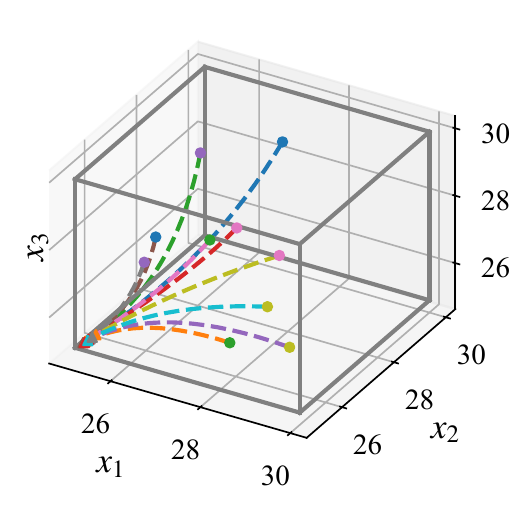}
		\caption{Constant {feasible} controller}
	\end{subfigure}
	\begin{subfigure}[t]{0.30\linewidth}
		\centering\includegraphics[width=.83\linewidth]{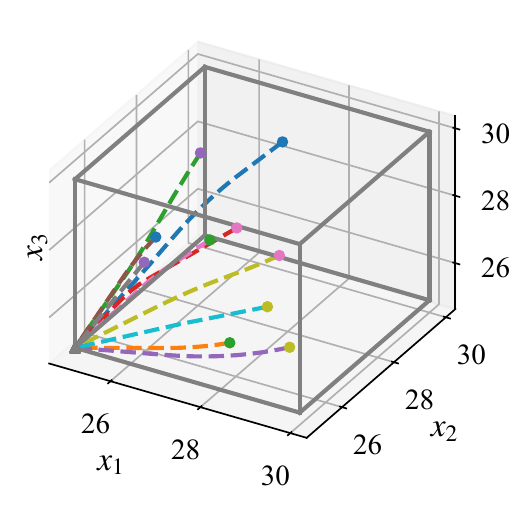}
		\caption{QP-based safety filter}
	\end{subfigure}
	\caption{ Three-room temperature trajectories using different controllers. }
	\label{fig:case2_result}
\end{figure*}

\emph{Case 3}. Consider the LTI system
\[
\dot{x}_1=x_2,\qquad 
\dot{x}_2=0.1x_1-0.1x_2+u,\qquad u\in[-1,1],
\]
with CBFs
\(
h_1(x)=1+x_1+x_2,\) \(
h_2(x)=1-x_1-x_2.
\)
Using $\alpha(h)=h$ gives
\[
\Psi=\left[\begin{smallmatrix}
1\\[1pt]
-1
\end{smallmatrix}\right],
\qquad
\delta=\left[\begin{smallmatrix}
1.1x_1+1.9x_2+1\\[1pt]
-1.1x_1-1.9x_2+1
\end{smallmatrix}\right].
\]
For any $x$, let $U_{\mathrm{safe}}(x)$ be the set of inputs satisfying both CBF
constraints at $x$. Two representative vertices give
\(
U_{\mathrm{safe}}(0,1)=[-1,-0.9],\) \( 
U_{\mathrm{safe}}(0,-1)=[0.9,1],
\)
so no single input is feasible at both.
{Since $\Psi$ is constant on
\(
H=\mathrm{co}\{(-1,0),(-1,1),(0,-1),(0,1),(1,0),(1,-1)\},
\)
the pairwise condition~\eqref{eq:pairwise_monotone} holds automatically.   {Together with Assump.~\ref{assum:standing-rev}, Cor.~\ref{cor:pairwise-mono-constant} applies. Thus, any convex combination of feasible vertex inputs
$u_\lambda=\sum_j \lambda_j u^j$
is feasible at the state $x=\sum_j \lambda_j x^j$.
This certifies compatibility over $H$.}

{To find feasible vertex inputs \(u^j\), we solve the joint LP~\eqref{eq:joint-pairwise-lp}, which returns
\(
u(-1,0)=0.1,\; u(-1,1)=0,\; u(0,1)=-0.9,\;
u(0,-1)=0.9,\; u(1,0)=-0.1,\; u(1,-1)=0.
\)
Blending these inputs yields a feasible control for all \(x\in H\).}

We obtain an explicit {feasible input} via convex interpolation (Thm.~\ref{thm:affine-interp}).
With $u_{\mathrm{des}}=0$, the polytope $H$ partitions into three regions (Fig.~\ref{fig:partition}).
On each region, the safety-filter optimizer is a convex combination of the CBF-QP solutions at the vertices.
As shown in Fig.~\ref{fig:2d-trajs}(a,b), trajectories from 10 random initial states remain in the safe region and inputs stay within bounds.
Across all tests, the convex-interpolation controller matches the online CBF-QP while maintaining safety throughout $H$.

\begin{figure}
    \centering
    \includegraphics[width=0.7\linewidth]{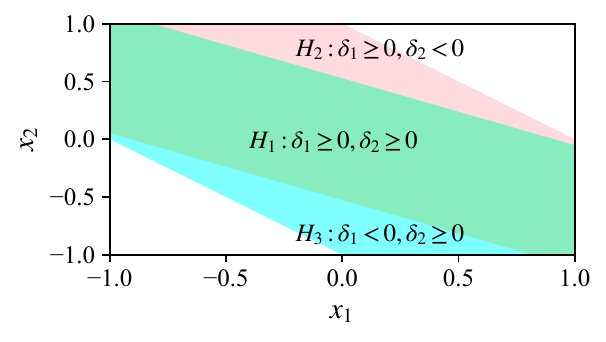}
    \caption{Partition of the polytope \(H\) into critical regions: green region \(H_1\) (both constraints inactive), pink region \(H_2\) (second constraint active), and cyan region \(H_3\) (first constraint active).}
    \label{fig:partition}
\end{figure}

\begin{figure}
    \centering
    \begin{subfigure}[t]{0.48\linewidth}
		\includegraphics[width=\linewidth]{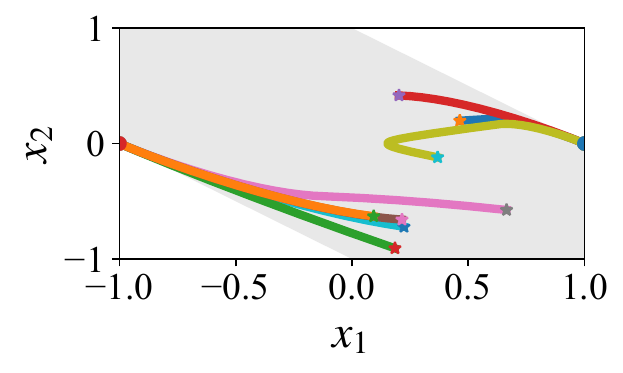}
		\caption{ Safe trajectories } 
	\end{subfigure}
	\begin{subfigure}[t]{0.48\linewidth}
		\centering\includegraphics[width=\linewidth]{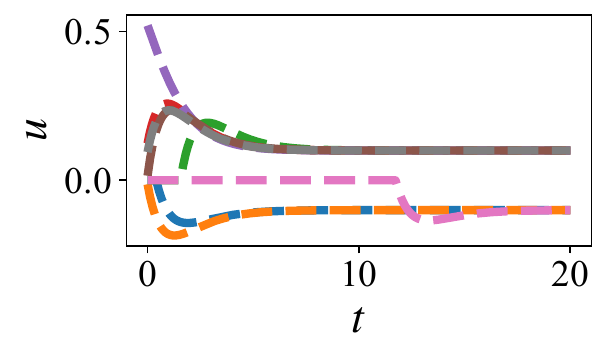}
		\caption{{Feasible} inputs}
	\end{subfigure}
    \caption{Trajectories and inputs under affine interpolation of the safety filters. Stars mark random initial states.}
    \label{fig:2d-trajs}
\end{figure}

\section{Conclusion}

 We proposed three \emph{Compatibility Propagation Certificates (CPCs)} that give lightweight, set-wise guarantees for extending vertex compatibility of CBF constraints to their convex hull. We also showed that under mild structural conditions, the safety filter admits an explicit affine form, enabling convex-combination evaluation without online optimization. {In practice, vertices can be chosen to capture critical operating conditions, and larger domains can be decomposed into multiple overlapping convex hulls. Systematic vertex selection and hull construction are key directions for future work.

\bibliographystyle{IEEEtran}
\bibliography{references}
\end{document}